\documentclass{amsproc}
\usepackage{graphicx}
\usepackage{color}
\newtheorem{theorem}{Theorem}[section]
\newtheorem{lemma}[theorem]{Lemma}
\newtheorem{proposition}[theorem]{Proposition}
\newtheorem{corollary}[theorem]{Corollary}

\theoremstyle{definition}
\newtheorem{definition}[theorem]{Definition}

\theoremstyle{remark}
\newtheorem{remark}[theorem]{Remark}

\numberwithin{equation}{section}

\def\C{\mathbb{C}}
\def\N{\mathbb{N}}
\def\R{\mathbb{R}}

\def\DD{\mathcal{D}}
\def\EE{\mathcal{E}}
\def\HH{\mathcal{H}}
\def\OO{\mathcal{O}}
\def\SS{\mathcal{S}}

\def\dom{{\rm dom\,}}
\def\id{{\rm id\,}}
\def\Ran{{\rm Ran\,}}
\def\spect{{\rm spect\,}}
\def\tr{{\rm tr\,}}

\def\nid{\noindent}
\def\vsth{\vskip2mm}

\begin{document}

\title{On the skeleton method and an application to a quantum scissor}

\author{H.D. Cornean}
\address{Department of Mathematical Sciences, Aalborg University, 
Fredrik Bajers Vej 7G, DK-9220 Aalborg, Denmark}

\email{cornean@math.aau.dk}
\thanks{The first author was supported in part by the Danish F.N.U. grant 
{\it Mathematical Physics and Partial Differential Equations}}

\author{P. Duclos}
\address{Centre de Physique Th\'eorique de Marseille UMR
  6207 - Unit\'e Mixte de Recherche du CNRS et des Universit\'es
  Aix-Marseille I, Aix-Marseille II et de l' Universit\'e du Sud
  Toulon-Var - Laboratoire affili\'e \`a la FRUMAM}
\email{duclos@univ-tln.fr}

\author{B. Ricaud}
\address{Centre de Physique Th\'eorique de Marseille UMR
  6207 - Unit\'e Mixte de Recherche du CNRS et des Universit\'es
  Aix-Marseille I, Aix-Marseille II et de l' Universit\'e du Sud
  Toulon-Var - Laboratoire affili\'e \`a la FRUMAM}
\email{ricaud@cpt.univ-mrs.fr}

\date{December 20, 2007 and, in revised form ...}



\begin{abstract}
 In the spectral analysis of few one dimensional quantum 
particles interacting through delta potentials it is well known that one
can recast the problem into the spectral analysis of an integral
operator (the skeleton) living on the submanifold which supports the
delta interactions. We shall present several tools which allow direct
insight into the spectral structure of this skeleton. 
We shall illustrate the method on a model of
a two dimensional quantum particle interacting with two infinitely
long  straight wires which cross one another at angle $\theta$: 
the quantum scissor. 
\end{abstract}

\maketitle

\section{Introduction}%
Let us consider the
following one dimensional model of $N$ quantum particles interacting 
through delta potentials. In suitable units, the corresponding 
Hamiltonian reads as
\begin{equation}\label{NPartIntThroughDeltaPotentials}
-\sum_{i=1}^N {\Delta_i\over 2m_i}+\sum_{1\le i<j\le N}Z_iZ_j\delta(x_i-x_j),
\quad \mbox{ acting in $L^2(\R^N)$,}
\end{equation} 
where $m_i$ and $Z_i$ denote respectively the mass and  the charge of the 
$i$'th particle. When the particles are identical (i.e. all the $m_i$'s and 
$Z_i$'s are equal), it is a well known fact that this model is indeed 
exactly solvable \cite{LL, McG}; for a quick and fairly complete
review, see \cite{vD}; see also the introduction of \cite{AlGH-KH-E}. However, it is not known whether 
the model is exactly solvable if the particles are distinct, 
and we strongly suspect that it is not. We have shown in \cite{CDR1}
that one can nevertheless expect 
partial exact results, at least. To explore this eventual solvability 
we have developed a mathematical tool, that we call the 
\textit{ skeleton method}, which requires to work with a system of
integral operators. 

The main issue of this article is to give a thorough exposition of
this skeleton method, see sections~\ref{sectSkeleton} and
\ref{SectionTTheta}. Finally we shall demonstrate the power of this
tool by the spectral analysis of bound states in a model of leaky
wires that we call a \textit{quantum scissor}, see \cite{BEPS} for this terminology.

\subsection{Leaky wires}%

We shall consider the problem (\ref{NPartIntThroughDeltaPotentials}) only  in the case $N=3$,  and 
\begin{equation}\label{MassChargeConditions}
m_1=m_2>0,\quad Z_1=Z_2<0\quad{\rm and}\quad  Z_3>0
\end{equation}
with the center of mass removed. Then the Hamiltonian expressed in the relative Jacobi coordinates acts in $L^2(\R^2)$. After rescaling ( see \cite{CDR1} for more details) we have
\begin{equation}\label{RescaledThreePartHamiltonian}
H:=-{1\over 2}\Delta_x-{1\over 2}\Delta_y-\delta( A_1^\perp\cdot (x,y))-\delta( A_2^\perp\cdot (x,y))+\lambda \delta( A_3^\perp\cdot (x,y))
\end{equation}
where $A_i$, $i=1,2,3$ are three normalized vectors as shown in Figure~1
\begin{figure}[tb]
\includegraphics[scale=0.5]{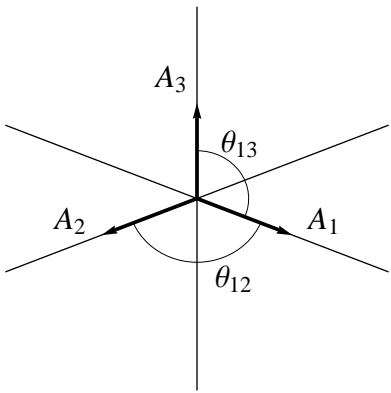}
\caption{The three supports of the $\delta$ leaky wires.}
\label{cdr3Fig1}
\end{figure}
where the angles $\theta_{i,j}$'s and $\lambda\ge 0$ depend on the original parameters $m_i$'s and $Z_i$'s.    Here $A_i^\perp$ denotes  $A_i$ rotated clockwise by $\pi/2$ and  the dot in $A_i^\perp\cdot (x,y)$ stands for the scalar product in $\R^2$. Thus $H$ in
(\ref{RescaledThreePartHamiltonian}) may be interpreted as the
Hamiltonian of a quantum particle confined to a two dimensional plane,
which interacts with three straight and infinitely long leaky wires
directed by the vectors $A_i$.  The "leaky wire" expression appears
probably for the first time in \cite{EI} . Another suitable expression  for such a quantum model is "leaky graph" which appears in \cite{EN}.

\subsection{Physical applications}%
Hamiltonians of the type (\ref{NPartIntThroughDeltaPotentials}) are
not 
only convenient mathematical models, but they do also describe 
physical systems when some physical parameters are pushed to a limit. 
It has been recognized long time ago, see e.g. \cite{Spr}, 
that atoms in a strong homogenous magnetic field can be modelled by 
(\ref{NPartIntThroughDeltaPotentials}), see \cite{BaSoY, BD} for a
recent mathematical treatment of this problem. Quasiparticles on
carbon nanotubes like excitons can be modelled by a system of charged 
quantum particles living at the surface of an infinitely long
cylinder, see \cite{P}. When the radius of the cylinder tends to zero,
it has been shown in \cite{CDP,CDR2} that a model of the type 
(\ref{NPartIntThroughDeltaPotentials}) { is a good effective
  Hamiltonian} for  these quasiparticles. Not only does the quantum
world provide us with such models. For example, in classical optics, photonic crystals  with a 
high contrast in the dielectric constant between the (thin) crystal
and air, can also be modelled by such a Hamiltonian, 
see \cite[\S 2]{KK} for more details.

\section{The skeleton}\label{sectSkeleton}%
Most of the content of this section could be obtained as a by-product of \cite{BEKS}. However we think it is worth to make public this more "operator theoretical" version.  For any normalized vector $A$ in $\R^2$ we introduce
$\tau_{A}:\HH^1(\R^2)\to L^2(\R)$ as the continuous restriction map 
\begin{equation}\label{defTauA}
\HH^1(\R^2)\ni \psi \mapsto \tau_A\psi\in L^2(\R), \quad \tau_A\psi(s):=\psi (sA).
\end{equation}
Let $g$ be a diagonal $3\times 3$ matrix with the diagonal entries
$\{g_i\}_{i=1}^3:=(-1,-1,\lambda)$. The Hamiltonian $H$ in 
(\ref{RescaledThreePartHamiltonian}) is
properly defined as the unique self-adjoint operator associated to the
closed and bounded from below quadratic form:
\begin{equation}\label{QHforH}
\HH^1(\R^2)\ni u\to {1\over 2}\|\nabla u\|^2+\sum_{i=1}^3g_i\|\tau_{A_i} u\|^2\in\R,
\end{equation}
Let us set $\tau_i:=\tau_{A_i}$ and $\tau:=(\tau_1,\tau_2,\tau_3):\HH^1(\R^2)\to \oplus_{i=1}^3 L^2(\R)$. Then $H$ may be rewritten as
\begin{equation}\label{RelativeMotionThreeParticleHamiltonian}
H=H_0+\tau^\star g\tau,\quad H_0:=-{1\over 2}\Delta,\quad \dom H_0:=\HH^2(\R^2).
\end{equation}
Notice that  the above sum defining $H$ must be understood in the sense of quadratic forms, and as a matter of fact $\dom H\ne\dom H_0$. Thanks to the particular values of the coupling constant $g_i$'s and by an application of the HVZ theorem one gets

\begin{lemma}\label{lemEssSpect}
For all $\lambda\ge -1$, the essential spectrum of $H$ is $[-{1\over2},\infty)$.
\end{lemma}%

We want to show that the eigenvalue problem $H\Psi=E\Psi$ for $E<-{1\over2}$, i.e. 
below the essential spectrum can be reduced to a one-dimensional eigenvalue problem
involving integral operators. Using Krein's formula  with $R(z):=(H-z)^{-1}$,
$R_0(z):=(H_0-z)^{-1}$ we get at once: 
\begin{equation}\label{KreinsFormula}
R(z)=R_0(z)-R_0(z)\tau^\star(g^{-1}+\tau R_0(z)\tau^\star)^{-1}\tau
R_0(z),\quad z\in\rho(H_0)\cap\rho(H).
\end{equation}
By classical Sobolev trace theorems the following operators are continuous:
$$
\tau R_0(z):L^2(\R^2)\to\bigoplus_{i=1}^3\HH^{3\over2}(\R),
\quad 
\tau R_0(z)\tau^\star:\bigoplus_{i=1}^3 \HH^{s}(\R) \to \bigoplus_{i=1}^3\HH^{s+1}(\R)
$$
for all $z\notin\spect H_0$ and all $s\in\R$. This allows to consider $g^{-1}+\tau R_0(z)\tau^\star$ as a bounded operator on $\SS:=\oplus_{i=1}^3 L^2(\R)$ when $z\notin \R_+$.
\begin{definition}\label{SkeletonDef}%
We shall call $S(k):=g^{-1}+\tau R_0(-k^2)\tau^\star$ the {\em skeleton} of $H$ at energy $-k^2$.
\end{definition}%
\begin{theorem}\label{HtoS}%
$E<-{1\over2}$ is an eigenvalue of $H$ iff $\ker(  g^{-1}+\tau
R_0(E)\tau^\star)\ne\{0\}$. If $P$ is the orthogonal projector on this
kernel, then the multiplicity of $E$ is equal to the dimension of $P$.
In addition, the operator $P\tau R_0^2(E)\tau^\star P$ is invertible
on the range of $P$, and the eigenprojector of $H$ associated to $E$ is given by
$$
R_0(E)\tau^\star\left(P\tau R_0^2(E)\tau^\star P\right)^{-1}\tau R_0(E).
$$
\end{theorem}%

\begin{proof} 
1. We start by showing that $\forall E:=k^2<-{1\over2}$ the essential spectrum of $S(k)$ obeys: for all $\lambda\ge0$ 
$$
0\notin \spect_{\rm ess}  S(k)=\spect_{\rm ac}  S(k)=[-1,-1+{1\over\sqrt{2}k}]\cup[\lambda^{-1},\lambda^{-1}+{1\over\sqrt{2}k}].
$$
Indeed if one sets
$$
T_{\theta_{i,j}}:= \tau_i R_0(z)\tau_j^\star
$$ then
\begin{equation}\label{leSquelette}
S(k)=
\begin{pmatrix}
-1+T_0 & 0    & 0\cr
                                     0   &-1+T_0 & 0\cr
                                     0   & 0    &\lambda^{-1}+T_0\cr
\end{pmatrix}
                                     +
\begin{pmatrix} 0             & T_{\theta_{1,2}}  & T_{\theta_{2,3}} \cr
    T_{\theta_{1,2}}   &0                  & T_{\theta_{2,3}} \cr
 T_{\theta_{2,3}}   &  T_{\theta_{2,3}}    &0\cr
\end{pmatrix} 
\end{equation} 
Since the diagonal of the first matrix consists of multiplication
operators (in the Fourier
representation see (\ref{detT0})) and the entries of the second matrix 
are all trace class operators (see Theorem~\ref{ROODTh}), we are
done. That $0\notin \spect_{\rm ess} S(k)$ is now obvious.

\vskip1mm
\nid 2. Assume that $E<-{1\over2}$ is an eigenvalue of
$H$, but $0$ is not an
eigenvalue of $S(k)$. Then $S(k)$ has a bounded inverse 
(after an easy application of the Fredholm alternative). Since
$R_0(E)$ and $\tau R_0(E)$ are bounded
operators, it means that $R(z)$ is also bounded at $z=E$. This 
contradicts the fact that $E$ is an eigenvalue of $H$. We conclude
that $S(k)$ cannot be invertible (injective) if $-k^2$ coincides with
an eigenvalue of $H$.

\vskip1mm
\nid 3. Now let us prove that all singularities of $S(k)^{-1}$
  correspond to eigenvalues of $H$. One has the identity
\begin{equation}\label{BWId}
(g^{-1}+\tau R_0(z)\tau^\star)^{-1}=g-g\tau R(z)\tau^\star g
\end{equation}
valid for $z$ where at least one and therefore two members of this identity exists. Now
assume that for some
$E<-1/2$, the operator 
$g^{-1}+\tau R_0(E)\tau^\star$ is not invertible (i.e. not injective in
our case). Assume also that $E$ is not in the (discrete) spectrum of
$H$. Then (\ref{BWId}) implies that in a small disc around $E$ we
have that $(g^{-1}+\tau R_0(z)\tau^\star)^{-1}$ is uniformly bounded,
which means that $g^{-1}+\tau R_0(E)\tau^\star$ is invertible by Neumann
series, contradiction. 
\vskip1mm
\nid 4. Now let us investigate the dimension of the spectral
  subspace associated to an eigenvalue.  Assume that $0$ is an eigenvalue of $g^{-1}+\tau
R_0(E)\tau^\star$ and let $P$ be the finite dimensional associated 
eigenprojector. We have shown that $E$ is also an eigenvalue of $H$,
and denote by $P(E)$ its finite dimensional projection. We want to
prove here that ${\rm dim}(P)= {\rm dim}(P(E))$. 

Since $(g^{-1}+\tau R_0(E)\tau^\star)P=0$ and using the resolvent identity:
\begin{equation}\label{pierre1}
(g^{-1}+\tau R_0(z)\tau^\star)P=(z-E)\tau R_0(z)R_0(E)\tau^\star P.
\end{equation}
Using (\ref{BWId}), and knowing that near $E$ we have
\begin{equation}\label{pierre2}
(z-E)R(z)=-P(E)+ {\mathcal O}((z-E)),
\end{equation}
it follows that 
\begin{align}\label{pierre3}
P&=(z-E)(g^{-1}+\tau R_0(z)\tau^\star)^{-1}\tau R_0(z)R_0(E)\tau^\star P \\&=
g  \tau P(E)\tau^\star  g \tau R_0(z)R_0(E)\tau^\star P + {\mathcal O}((z-E)).\nonumber 
\end{align}
Taking the limit $z\to E$ we obtain
\begin{equation}\label{pierre4}
P=g  \tau P(E)\tau^\star  g \tau R_0(E)^2\tau^\star P .
\end{equation}
If ${\rm Ran}(P(E))$ is spanned by the eigenvectors $\{\psi_j\}_{j=1}^{{\rm
    dim}(P(E))}$, then (\ref{pierre4}) says that ${\rm Ran}(P)$ is
    spanned by  $\{g\tau \psi_j\}_{j=1}^{{\rm
    dim}(P(E))}$, therefore 
$${\rm dim}(P)\leq  {\rm dim}(P(E)).$$

We now want to prove the reverse inequality. Denote by 
$Q:={\rm id}-P$, so that in matrix notation (use (\ref{pierre1}) and
its adjoint)
$$
g^{-1}+\tau R_0(z)\tau^\star 
=(z-E)\left(\begin{array}{cc}
P\tau R_0(z)R_0(E)\tau^\star P&P\tau R_0(z)R_0(E)\tau^\star Q\cr
Q\tau R_0(z)R_0(E)\tau^\star P&(z-E)^{-1}Q(g^{-1}+\tau R_0(z)\tau^\star) Q \cr
\end{array}
\right).
$$
To invert this matrix we use the Feshbach method. One has
(i) the operator $Q(g^{-1}+\tau R_0(E)\tau^\star) Q$ has a bounded
inverse on the range of $Q$, thus $(z-E)^{-1}Q(g^{-1}+\tau R_0(z)\tau^\star) Q$ is bounded invertible for $z$ in a
neighbourhood of $E$, except eventually at $E$, and (ii) 
the following operator is
bounded invertible at least in a neighbourhood of $E$
$$
A(z):= P\tau R_0(z)R_0(E)\tau^\star P- P\tau
R_0(z)R_0(E)\tau^\star Q(z-E)$$
$$\cdot (Q(g^{-1}+\tau
R_0(z)\tau^\star) Q)^{-1}Q\tau R_0(z)R_0(E)\tau^\star P.
$$
Notice that this operator is nothing but a finite dimensional matrix,
acting in ${\rm Ran}(P)$. 
The Feshbach formula says that the above operator $A(z)$ is invertible if and
only if $( g^{-1}+\tau R_0(z)\tau^\star )/(z-E)$ is invertible. Moreover,
for $z$ in a neighborhood of $E$ 
this formula gives:
$$ (z-E)P (g^{-1}+\tau R_0(z)\tau^\star)^{-1} P= A(z)^{-1}, \quad
z\neq E.$$
Using again (\ref{BWId}) and (\ref{pierre2}), we obtain 
$$ A(z)^{-1}= Pg\tau P(E)\tau^\star  g P +{\mathcal O}(z-E), \quad
z\neq E.$$
This inverse is bounded near $E$, and $A(z)$ is continuous at $z=E$,
hence $A(E)$ is invertible and 
\begin{equation}\label{pierre7}
A(E)^{-1}=\left \{ P \tau R_0(E)^2\tau^\star  P\right \}^{-1}= Pg\tau
P(E)\tau^\star  g P.
\end{equation}
Summarizing, via the Feshbach formula, we obtain that 
\begin{equation}\label{pierre8}
(z-E) (g^{-1}+\tau R_0(z)\tau^\star)^{-1} = A(z)^{-1}+{\mathcal O}((z-E)).
\end{equation}
Multiply (\ref{BWId}) with $(z-E)$, use (\ref{pierre2}), (\ref{pierre8}), and take the
limit $z\to E$. This gives:
\begin{equation}\label{pierre9}
P(E)= R_0(E)\tau^\star  A(E)^{-1}\tau  R_0(E)=R_0(E)\tau^\star  Pg\tau
P(E)\tau^\star  g P\tau R_0(E).
\end{equation}
Now assume that $\{\phi_j\}_{j=1}^{{\rm dim}(P)}$ are eigenvectors
spanning the range of $P$. Then (\ref{pierre9}) says that the range of
$P(E)$ is spanned by $\{R_0(E)\tau^\star \phi_j\}_{j=1}^{{\rm dim}(P)}$,
which implies 
$${\rm dim}(P(E))\leq  {\rm dim}(P)$$
and we are done. 
\end{proof}

\section{The $T_\theta$ operators}\label{SectionTTheta}%
In this section we shall establish various properties of the 
$T_{\theta_{i,j}}$ operators. 

\subsection{Generalities}%

Let $A$ and $B$ be two normalized vectors of $\R^2$. We shall consider
$
\tau_A R_0(-k^2)\tau_B^\star
$
where $\tau_A$, $\tau_B$ are defined by (\ref{defTauA}). We can
  obtain explicit formulas for their integral kernels using the
  Fourier transform that we denote by a hat. 
We summarize the results in the following technical lemma:

\begin{lemma}%
The operator $\tau_A R_0(-k^2)\tau_B^\star$ depends only on the angle 
$\theta$ between the vectors $A$ and $B$. When $\det(A,B)\ne 0$, 
the Fourier transform of $\tau_A R_0(-k^2)\tau_B^\star$ is an integral 
operator with kernel
\begin{equation}\label{detTthetaKernel}
\hat T_\theta(t,s;k)={1\over2\pi|\sin(\theta)|}{1\over
{t^2-2\cos(\theta) ts +s^2\over 2\sin^2(\theta)}+k^2}.
\end{equation}
When $A=B$, the Fourier transform of  $\tau_A R_0(-k^2)\tau_A^\star$ 
is the multiplication operator given by the function
\begin{equation}\label{detT0}
\hat T_0(s;k):={1\over\sqrt{s^2+2k^2}}.
\end{equation}
\end{lemma}%
\vsth
\nid The proof of this lemma is elementary and left to the reader. 
\vsth
\begin{remark}\label{remTTheta}%
\ (a) One has:
$$
\forall \theta\in(-\pi,\pi),\quad \|T_\theta(k)\|\le{1\over\sqrt{2}k}.
$$
This is clear for the case $\theta=0$ since then $\hat T_0(k)$ is an explicit multiplication operator. For $\theta\ne 0$ we use
$$
\|\tau_A R_0\tau_B^\star\|^2\le\|\tau_A R_0^{1\over2}\|^2\|R_0^{1\over2}\tau_B^\star\|^2=
\|\tau_A R_0\tau_A^\star\|\|\tau_B R_0\tau_B^\star\|=\|T_0\|^2.
$$
One can also compute explicitly the Hilbert-Schmidt norm of $T_\theta$:
$$
\|T_\theta(k)\|_{\rm HS}^2={1\over 2\pi\sin(\theta) k^2},\quad \theta\ne0\ \mbox{mod}\, \pi .
$$
\vskip1mm
\nid(b) If we perform the scaling $s\to ks$ then clearly $\hat T_\theta(k)$ becomes $k^{-1}\hat T_\theta(1)$. Since in the sequel we shall use this property and work only with $\hat T_\theta(1)$, we  denote
$$
\hat T_\theta:=\hat T_\theta(1).
$$
\vskip1mm
\nid(c) Let $\Pi:L^2(\R)\to L^2(\R)$ denote the parity operator $\Pi\varphi(s)=\varphi(-s)$. Then 
$
[\Pi,\hat T_\theta]=0
$
so that one can decompose
$$
\hat T_\theta=\hat T_\theta^+\oplus \hat T_\theta^-,\quad {\rm with}\quad 
\hat T_\theta^\pm:= {1\pm\Pi\over2}\hat T_\theta.
$$
\vskip1mm
\nid(d)  By a simple inspection of  (\ref{detTthetaKernel}) we have the reflection properties:
$$
\forall \theta\in(0,\pi),\quad T_{\pi-\theta}^\pm=\pm T_{\theta}^\pm
$$
\end{remark}%

\subsection{Rank one operator decomposition of $\hat T_\theta$}%
Let us first consider  $\hat T_{\pi\over2}$; we have the formula
\begin{equation}\label{BrummelhuisFormula}
\hat T_{\pi\over2}(p,q)={1\over \pi}{1\over p^2+q^2+2}={1\over\pi}\int_0^\infty e^{-2s} e^{-sp^2}e^{-sq^2}ds
\end{equation}
which shows that $T_{\pi\over2}$ is a "sum" of positive rank one operators so that $\hat T_{\pi\over2}\ge0$. Since $(p,q)\to T_{\pi\over2}(p,q)$ is continuous and 
$$
 \int_\R \hat T_{\pi\over2}(p,p)dp={1\over2}<\infty
$$
this shows in view of \cite[th\,2.12]{Si}, that $T_{\pi\over2}$ is trace class and that its trace and therefore its trace norm are $1/2$. We are indebted to R. Brummelhuis who showed us the trick (\ref{BrummelhuisFormula}). The above derivation can be generalized to any angle $0<\theta<\pi$ as follows.

\begin{theorem}\label{ROODTh} %
For all $\theta\in(0,\pi)$ one has  in the trace norm ( $\|\cdot\|_1$) topology
\begin{eqnarray}
\hat T_\theta&=&2^{-{1\over2}}{\sin(\theta)\over\pi}\sum_{n\in\N}
{\Gamma(n+{1\over2})\over \Gamma(n+1)} \cos^n(\theta)\int_0^\infty ds\,s^{-{1\over2}}
e^{-2\sin^2(\theta) s}  P_{n,s}\label{RODofTTheta}
\\
\hat T_\theta^+
&=&
2^{-{1\over2}}{\sin(\theta)\over\pi} \sum_{n=0}^\infty {\Gamma(2n+{1\over2})\over \Gamma(2n+1)} \cos^{2n}(\theta)\int_0^\infty ds\, s^{-{1\over2}}
e^{-2\sin^2(\theta) s} P_{2n,s} 
\label{decTThetaPlusEnRang1}\nonumber\\
\hat T_\theta^-
&=&
2^{-{1\over2}}{\sin(\theta)\over\pi} \sum_{n=0}^\infty 
{\Gamma(2n+{3\over2})\over \Gamma(2n+2)} 
\cos^{2n+1}(\theta)\int_0^\infty ds s^{-{1\over2}}
e^{-2\sin^2(\theta) s}P_{2n+1,s} \nonumber
\end{eqnarray}
where $P_{n,s}$ denotes the rank one orthogonal projector on the vector $g_{n,s}$ defined by 
\begin{equation}\label{vectorgns}
g_{n,s}(p):=\sqrt{(2s)^{n+{1\over2}}\over\Gamma(n+{1\over2})} p^n e^{-p^2 s}.
\end{equation}
Accordingly one has
\begin{equation}\label{SignOfTTheta}
\forall \theta\in(0,{\pi\over2}],\quad \hat T_\theta^\pm\ge0\quad{\rm and}\quad
\forall \theta\in[{\pi\over2},\pi),\quad \pm\hat  T_\theta^\pm\ge0.
\end{equation}
It follows that $T_\theta$ and $T_\theta^\pm$ are trace class and
\begin{eqnarray*}
\tr T_\theta^+&=&\|T_\theta^+\|_1= 
\frac{ \cos \left(\frac{\theta }{2}\right)+\sin
   \left(\frac{\theta }{2}\right)}{2 \sqrt{2}\sin(\theta)}\\
   \tr T_\theta^-&=& 
\frac{ \cos \left(\frac{\theta }{2}\right)-\sin
   \left(\frac{\theta }{2}\right)}{2 \sqrt{2}\sin(\theta)},\quad \|T_\theta^-\|_1=
   \frac{ \left|\cos \left(\frac{\theta }{2}\right)-\sin
   \left(\frac{\theta }{2}\right)\right|}{2 \sqrt{2}\sin(\theta)}
   \\
\tr T_\theta&=& 
{1\over 2\sqrt{2}\sin\left(\theta\over2\right)},\quad \|T_\theta\|_1=
\frac{ \max\left\{\cos \left(\frac{\theta }{2}\right),\sin
   \left(\frac{\theta }{2}\right)\right\}}{ \sqrt{2}\sin(\theta)}
\end{eqnarray*}
\end{theorem}%

\begin{proof} To find the rank one operator decomposition of $\hat T_\theta$ we simply expand its kernel as follows. Let $A:= p^2+q^2+2\sin^2\theta$ and $B:=2pq \cos(\theta)$, one can easily check  that $A>0$ and $|B/A|<1$ for all $0<\theta<\pi$. Thus one has
\begin{eqnarray*}
&&
{1\over2\pi\sin(\theta)}{1\over{p^2-2\cos(\theta) pq +q^2\over 2\sin^2(\theta)}+1}
=
{\sin\theta\over\pi} {1\over A-B}
=
{\sin(\theta)\over\pi} \sum_{n=0}^\infty A^{-1}\left(B\over A\right)^n
\\
&=&{\sin(\theta)\over\pi}\sum_{n=0}^\infty B^n\int_0^\infty  ds{s^ne^{-sA}\over n!}
\\
&=&{\sin(\theta)\over\pi}\sum_{n=0}^\infty 2^n{\cos^n(\theta)\over n!}
\int_0^\infty ds\, e^{-2s\sin^2(\theta)}s^ne^{-s(p^2+q^2)} (pq)^n.
\end{eqnarray*}
To arrive at (\ref{RODofTTheta}) one needs to normalize in $L^2(\R)$ the vector $p\to p^ne^{-s p^2}$ which gives the vector $g_{n,s}$ in (\ref{vectorgns}). Since $\|P_{n,s}\|=\|P_{n,s}\|_1$ the convergence in the trace norm topology of the r.h.s. of (\ref{RODofTTheta}) is true since the terms in the following sum are all positive and one has explicitly:
$$
2^{-{1\over2}}{\sin(\theta)\over\pi}\sum_{n\in\N}
{\Gamma(n+{1\over2})\over \Gamma(n+1)} |\cos^n(\theta)|\int_0^\infty ds\,s^{-{1\over2}}
e^{-2\sin^2(\theta) s} ={1\over2\sqrt{1-|\cos(\theta)|}}.
$$
This shows that $\hat T_\theta$ is trace class and that (\ref{RODofTTheta}) is valid in the trace norm topology. Since $g_{n,s}$ has the parity of $n$, one  gets $\hat T_\theta^\pm$ by selecting the even and odd values of $n$ in (\ref{RODofTTheta}) resp.. The rest is now obvious up to some tedious explicit computations of sums.
\end{proof}
We shall draw some other useful  properties from the above theorem.

\begin{corollary}\label{corTTheta}%
(i) $\theta\to T_\theta$ is a selfadjoint analytic family as a map
from $\DD:=\{\theta\in\C,\, |\cos(\theta)|<1\}$ with values in the
ideal of trace
class operators.\vskip1mm
\nid
(ii) If one  labels the eigenvalues of $T_\theta^+$ by descending order:
$E_1^+(\theta)\ge E_2^+(\theta)\ge \ldots\ge E_n^+(\theta)\ge \ldots$ then each function $(0,\pi)\ni\theta\to E_n^+(\theta)$ is continuous and decreasing on $(0,{\pi\over2}]$ and increasing on $[{\pi\over2},\pi)$.  
\vskip1mm
\noindent
If one labels the eigenvalues of $T_\theta^-$ by descending order on $(0,{\pi\over2})$ and ascending order on  $({\pi\over2},\pi)$, then each function  $(0,\pi)\ni\theta\to E_n^-(\theta)$ is continuous and decreasing on $(0,\pi)$.
\end{corollary}%

\begin{proof} (i) is a direct consequence of the convergence of the r.h.s. of (\ref{RODofTTheta}) on $\DD$. To prove (ii) we shall consider another s.a. family of operators which is the image of  $\{T_\theta\}_\theta$ under the scaling $p\to\sin(\theta)p$, $0<\theta<\pi$:
\begin{equation}\label{defTThetaSharp}
{^\sharp T}_\theta(p,q)={1\over\pi}{1\over p^2+q^2-2pq\cos(\theta)+2}.
\end{equation}  Then proceeding as in the previous theorem we get
\begin{eqnarray*}
{^\sharp T}_\theta^+&=&{1\over\pi}\sum_{n\in\N} (2\cos\theta)^{2n} B_{2n}\\
{^\sharp T}_\theta^-&=&{1\over\pi}\sum_{n\in\N} (2\cos\theta)^{2n+1} B_{2n+1}
\end{eqnarray*} 
where $B_n$ denotes the positive operator with kernel
$$
B_n(p,q):={(pq)^n\over (p^2+q^2+2)^{n+1}}=\int_0^\infty {s^{n}\over n!}e^{-2s}(pq)^ne^{-sp^2}e^{-sq^2}. 
$$
If we label the eigenvalues of $T_\theta^+$, i.e. the eigenvalues of 
${^\sharp T}_\theta^+$ in descending order they are all continuous in 
$\theta$ and in view of the elementary dependence of ${^\sharp
  T}_\theta^+$ on $\theta$, they are decreasing on $(0,\pi/2]$ and 
increasing on $[\pi/2,\pi)$. We skip the analogous reasoning for $T_\theta^-$.
\end{proof} 
\subsection{$T_\theta^\pm$ are ergodic}%
The reader can find the definition of an ergodic operator in 
\cite[\S XIII.12]{RS4}.
\begin{proposition}\label{PropTThetaIsErgodic}%
(i) For all $\theta\in(0,\pi)$,  $ T^+_\theta$ is ergodic and
$\sup T_\theta^+$ is a simple eigenvalue of $T_\theta^+$.

\vskip1mm
\nid (ii) For all $\theta \in (0,\pi/2)\cup(\pi/2,\pi)$,  
$ {\rm sign}(\pi/2-\theta)T^-_\theta$ is
ergodic and $ \sup {\rm sign}(\pi/2-\theta) T_\theta^-$ is a simple 
eigenvalue of $T_\theta^-$. 

\end{proposition}%
\begin{proof}
We have previously seen that $T^+_\theta\ge0$. Also $T_\theta^+$  is
self adjoint and compact, thus $\|T_\theta^+\|=\sup T_\theta^+$ is an 
eigenvalue of $T^+_\theta$. Clearly since $T_\theta^+(p,q)>0$ for all 
$p$ and $q$, one has $(T_\theta^+f,g)>0$ whenever $f$ and $g$ are
positive. Thus $T_\theta^+$ is ergodic. By applying \cite[Th. XIII.43]{RS4}  we get that $\sup T_\theta^+$ is a simple eigenvalue of $T_\theta^+$. The proof for $T_\theta^-$ is analogous.
\end{proof}

\subsection{$T_\theta$ is injective}%
This question was brought to us by T. Dorlas. 
\begin{lemma}%
For all $0<\theta<\pi$, one has $\ker T_\theta=\{0\}$. 
\end{lemma}%
\begin{proof} We find it more convenient to work with ${^\sharp\!T}_\theta$, see (\ref{defTThetaSharp}),  which is unitarily equivalent to $T_\theta$. We recall that $\Pi$ denotes the parity operator, see Remark~\ref{remTTheta}(c). Using the formula derived in the proof of Corollary~\ref{corTTheta}, we get with $\varphi\in \Ran\Pi^+$, that ${^\sharp\!T}_\theta^+\varphi=0$ implies
$$
({^\sharp\!T}_\theta^+\varphi,\varphi)=0\quad\Rightarrow\quad
{1\over\pi}\sum_{n=0}^\infty {(2\cos(\theta))^{2n}\over (2n)! }\int_0^\infty s^{2n} e^{-2s}(C_{2n,s}\varphi,\varphi) ds=0
$$
where $C_{n,s}$ with kernel $C_{n,s}(p,q):=(pq)^n e^{-sp^2}e^{-sq^2}$ is a positive rank one operator. 
If $\theta\ne {\pi\over2}$ it follows that:
$
\forall n\in\N,\ \forall s>0,\ (C_{2n,s}\varphi,\varphi)=0 
$
since $(2\cos(\theta))^{2n}$ and $s^{2n} e^{-s}$ are strictly positive. But
$$
(C_{2n,s}\varphi,\varphi)=0 \iff \int_\R |p^{2n}e^{-sp^2}\varphi(p)|^2 dp=0
$$
which shows that $\varphi\perp p^{2n}e^{-{p^2\over 2}}$ for all $n\in\N$ by choosing $s=1/2$. Clearly $\{p^{2n}e^{-{p^2\over 2}}, n\in\N\}$ is total in $\Ran\Pi^+$ since they generate the even Hermite functions. Thus $\varphi=0$. A similar argument shows that if $\varphi\in\Ran\Pi^-$ and ${^\sharp\!T}_\theta^-\varphi=0$ then $\varphi=0$; notice that it is understood here that $\theta\ne{\pi\over2}$ since otherwise ${^\sharp\!T}_\theta^-=0$.

Finally we consider the case $\theta=\pi/2$ and $\varphi\in\Ran\Pi^+$. Here we get as above
$$
\forall s>0,\quad (\varphi, e^{-sp^2})=0 
$$
and by differentiating indefinitely this identity with respect to $s$
in $s=\frac{1}{2}$ we find
$$
\varphi\perp p^{2n}e^{-{p^2\over 2}},\quad \forall n\in\N,
$$ 
which implies as above that $\varphi=0$. 
\end{proof}

\subsection{Some properties of  $(2^{-{1\over2}}-\hat{T}_0)^{-1/2}
  \hat T_\theta^-
(2^{-{1\over2}}-\hat{T}_0)^{-1/2}$}%
From the rank one operator decomposition of $T_\theta^-$, see Theorem~\ref{ROODTh},  one gets
\begin{eqnarray}\label{defTTildeThetaMoins}
\widetilde T_\theta^-&:=&(2^{-{1\over2}}-\hat{T}_0)^{-1/2} \hat
T_\theta^-
(2^{-{1\over2}}-\hat{T}_0)^{-1/2}
\nonumber\\
&=&2^{-{1\over2}}{\sin(\theta)\over\pi} \sum_{n=0}^\infty 
{\Gamma(2n+{3\over2})\over \Gamma(2n+2)} 
\cos^{2n+1}(\theta)\int_0^\infty ds s^{-{1\over2}}
e^{-2\sin^2(\theta) s}\widetilde P_{2n+1,s} \label{ROODofTildeTTheta}
\end{eqnarray}
where 
$$
\widetilde P_{2n+1,s} :=(\cdot,\widetilde g_{2n+1,s} )\widetilde g_{2n+1,s}\quad{\rm with}\quad
\widetilde g_{2n+1,s}:=(2^{-{1\over2}}-\hat{T}_0)^{-1/2} g_{2n+1,s}.
$$
It turns out that $\widetilde g_{2n+1,s}$ belongs to $L^2(\R)$ and the r.h.s. (\ref{ROODofTildeTTheta}) is convergent in the trace norm. More precisely
\begin{eqnarray*}
\|\widetilde P_{2n+1,s} \|_1&=&\|\widetilde g_{2n+1,s}\|^2=\frac{\sqrt{2} \left(4 n+8 s+4 \sqrt{s} \,U\left(-\frac{1}{2},-2 n,4
   s\right)+1\right)}{4 n+1}\\
   &\le& \frac{8 n \sqrt{s} \,\Gamma (2 n)+(4 n+16 s+1) \Gamma \left(2
   n+\frac{1}{2}\right)}{\sqrt{2}\, \Gamma \left(2
   n+\frac{3}{2}\right)}
\end{eqnarray*}
where $U$ denotes the confluent hypergeometric function, see \cite[13.1.3]{AS}. The last estimate is obtained by integration of the r.h.s. of the bound 
$$
\left|\tilde g(2n+1,s)(p)\right|^2\le 
{(2s)^{2n+{3\over2}}\over\Gamma(2n+{3\over2})} 
p^{4n} e^{-2sp^2}(\sqrt{2} p^2+2 |p|+4 \sqrt{2}).
$$
which is more convenient in view of the summations over $s$ and $n$. Then the integration over $s$ gives
\begin{eqnarray*}
&&
\sin(\theta)\int_0^\infty 
\frac{e^{-2 s \sin ^2(\theta )} }{\sqrt{s}} {\Gamma(2n+{3\over2})\over \Gamma(2n+2)}  \|\tilde g(2n+1,s)\|^2 ds
\le \\
&&\frac{\sqrt{\pi }\Gamma \left(2 n+\frac{1}{2}\right) }{2 \Gamma (2 n+2)}
+
\frac{2 n \sqrt{\pi } \Gamma \left(2n+\frac{1}{2}\right)}{ \Gamma (2 n+2)}
+
\frac{2 \sqrt{\pi } \Gamma \left(2 n+\frac{1}{2}\right) }{\sin^2(\theta)\,\Gamma (2 n+2)}
+
\frac{\sqrt{2} \Gamma (2 n+1)}{ \sin(\theta)\,\Gamma (2 n+2)}\\
&=:& a_1+a_2+a_3+a_4.
\end{eqnarray*}
In the summation over $n$ of these four terms, only the second one causes problems to arrive at a convenient final formula: ( the sums are computed with $0<\theta<\pi/2$)
\begin{eqnarray*}
\sum_{n=0}^\infty {a_1\cos^{2n+1}(\theta)\over \sqrt{2}\pi}
&=&
\frac{1}{2} \left(\cos \left(\frac{\theta }{2}\right)-\sin \left(\frac{\theta }{2}\right)\right)
\\
\sum_{n=0}^\infty {a_2\cos^{2n+1}(\theta)\over \sqrt{2}\pi}   
&=&
\frac{\cos ^3(\theta ) \,_2F_1\left(\frac{5}{4},\frac{7}{4};\frac{5}{2};\cos ^2(\theta )\right)}{4 \sqrt{2}}
\\
\sum_{n=0}^\infty {a_3\cos^{2n+1}(\theta)\over \sqrt{2}\pi}     
&=&
{2\over \sin^2(\theta )} \left(\cos \left(\frac{\theta }{2}\right)-\sin\left(\frac{\theta }{2}\right)\right)\\
\sum_{n=0}^\infty {a_4\cos^{2n+1}(\theta)\over \sqrt{2}\pi}  
&=&
\frac{\tanh ^{-1}(\cos (\theta ))}{\pi (\sin(\theta ) }
\end{eqnarray*}
One replaces $a_2$ by the the following bound valid for all $n\in\N$: 
$$
\frac{2 n \Gamma \left(2 n+\frac{1}{2}\right)}{\Gamma (2 n+2)}\leq
   \frac{\Gamma \left(2 n+\frac{1}{2}\right)}{\Gamma (2 n+1)}=a_2'
$$
which gives
$$
\sum_{n=0}^\infty {a_2'\cos^{2n+1}(\theta)\over \sqrt{2}\pi}   =
\frac{1}{2} {\cos (\theta )\over \sin(\theta )} \left(\cos
   \left(\frac{\theta }{2}\right)+\sin \left(\frac{\theta
   }{2}\right)\right).
$$
Summing up gives
\begin{lemma}\label{boundOnTiltdeTThetaMoins}%
For all $\theta\in(0,\pi)$, $\widetilde T_\theta^-$ is trace class and for all $0<\theta<\pi/2$ one has:
$$
\|\tilde T_\theta^-\|_1\le \frac{\left(4\sin(\theta ) \tanh ^{-1}(\cos (\theta ))+\pi  \left(9 \cos \left(\frac{\theta }{2}\right)-\cos
   \left(\frac{5 \theta }{2}\right)-9 \sin \left(\frac{\theta
   }{2}\right)+\sin \left(\frac{5 \theta }{2}\right)\right)\right)}{4
   \pi  \sin(\theta ) }.
   $$
In particular:
\begin{equation}\label{TraceNormUpperBoundForTildeTTheta2PiOver3}
\|\widetilde T_{2\pi\over3}^-\|_1= \|\widetilde T_{\pi\over3}^-\|_1\le
-\frac{4}{3}+\frac{5}{\sqrt{3}}+\frac{\log (3)}{\sqrt{3} \pi }\sim 1.75532
\end{equation}
\end{lemma}%
\begin{remark}\label{remarkTildeTTheta2PiOver3}%
(a) If we do not replace $a_2$ by $a_2'$ we get a better bound:
$$
\|\tilde T_{2\pi\over3}^-\|_1\le 1.38929.
$$
Moreover a direct numerical evaluation on the Hilbert Schmidt norm gives
$$
\|\tilde T_{2\pi\over3}^-\|_{\rm HS}\sim 1.01327.
$$
\vskip1mm
\nid (b) We shall see below that $-1$ is an eigenvalue of $\widetilde T_{2\pi\over3}$, see (\ref{exactPbVp3}). Since the trace norm of $\widetilde T_{2\pi\over3}^-$ is less than 2, see (\ref{TraceNormUpperBoundForTildeTTheta2PiOver3}), it follows that this eigenvalue is simple and is the lowest eigenvalue of $\widetilde T_{2\pi\over3}$. Thus we may conclude that
\begin{equation}\label{infTildeTTheta2PiOver3}
\inf \widetilde T_{2\pi\over3}=-1.
\end{equation}
\vskip1mm
\nid (c)  The statements in Corollary~\ref{corTTheta}(ii) and Proposition~\ref{PropTThetaIsErgodic} for eigenvalues of $T_\theta^-$  works as well for those of $\widetilde T_\theta^-$. In particular the lowest one is simple and monotonically decreasing and pass by $-1$ for $\theta=2\pi/3$ in view of (\ref{infTildeTTheta2PiOver3}). 
\end{remark}%

\subsection{Exact eigenvalues and eigenvectors}\label{ExactEigenvaluesAndEigenvectorsSect}%
We collect here some exact results about these $\hat T_\theta$ operators. They can be checked by direct inspection.

\begin{equation}\label{exactPbVp1}
(\hat T_0+\hat T_{\pi\over2}) \varphi=\varphi,\quad{\rm with}\quad
\varphi(p)=\sqrt{2\over\pi}\,{1\over p^2+1},
\end{equation}
\begin{equation}\label{exactPbVp2}
(\hat T_0+2\hat T_{2\pi\over3}) \varphi=\sqrt{2}\varphi,\quad{\rm with}\quad
\varphi(p)=\frac{6^{3/4}}{\sqrt{\pi }\left(2 p^2+3\right) }
\end{equation}
and
\begin{equation}\label{exactPbVp3}
\widetilde{T}_{2\pi\over3}\varphi=-\varphi \quad
{\rm with}\quad\varphi(p)={\sqrt{T_0(0)-T_0(p)}\over p(2p^2+3)},\quad \|\varphi(p)\|^2=
\frac{1}{18} \left(6-\sqrt{3} \pi \right).
\end{equation}
(\ref {exactPbVp1}) and (\ref {exactPbVp2}) are simply obtained by
translating known exact eigenfunctions in the skeleton frame work; the
first one comes from the exactly solvable quantum scissor, see
(\ref{defHTheta}),  with angle $\pi/2$. The second one comes from the
Mc Guire bound state eigenfunction of its three particle system, see
\cite[\S IV.D]{McG}. The last one seems to be new. Since
$\sqrt{\hat{T}_0(0)-\hat{T}_0(p)} \sim 2^{-{5\over4}}|p|$ as $p\to0$,
this function has a cusp at 0.

\section{A quantum scissor}\label{sectionQS}
We consider the Hamiltonian (\ref{RelativeMotionThreeParticleHamiltonian}) in the particular case $\lambda=0$:
\begin{equation}\label{defHTheta}
H_\theta:=-{\Delta\over2}-\delta(A_1^\perp\cdot)-\delta(A_2^\perp\cdot)=-{\Delta\over2}-\tau_1\tau_1^\star-\tau_2\tau_2^\star 
\end{equation}
which described a two dimensional particle in a scissor-shaped waveguide, a name  borrowed from \cite{BEPS}. 
We assume without loss of generality that $\theta:=\theta_{1,2}$ belongs to
\begin{equation}\label{domTheta}
\theta\in[{\pi\over2},\pi);
\end{equation}
where $\theta$ denotes the angle made by the two vectors $A_1$ and $A_2$ which generate the supports of the delta interactions, see Figure~\ref{cdr3Fig1}. We note that the case $\theta=\pi$ is exactly solvable, and that the angles $\theta\in(0,\pi/2]$ are covered by the cases (\ref{domTheta}) since $H_\theta$ and $H_{\pi-\theta}$ are unitarily equivalent.

Thanks to Lemma~\ref{lemEssSpect}, the essential spectrum of $H_\theta$ is $[-1/2,\infty)$. The skeleton associated to $H_\theta$ in the Fourier representation is
$$
\begin{pmatrix}
-1+\hat T_0(k) & \hat T_\theta(k)\cr
\hat T_\theta(k) & -1+\hat T_0(k)\cr
\end{pmatrix}
\sim 
k^{-1}\begin{pmatrix}
-k+\hat T_0 & \hat T_\theta\cr
\hat T_\theta & -k+\hat T_0\cr
\end{pmatrix}
$$
where the unitarily equivalent second expression is obtained through the scaling $s\to ks$, see Remark~\ref{remTTheta}(b).  It acts on $\SS:=L^2(\R)\oplus L^2(\R)$. Thus, according to Theorem~\ref{HtoS},  $-k^2<-1/2$ is an eigenvalue of $H_\theta$ iff $k$ is an eigenvalue of
$$
{\bf T}_\theta :=\begin{pmatrix}
\hat T_0 & \hat T_\theta\cr
\hat T_\theta & \hat T_0\cr
\end{pmatrix}.
$$ 
Notice that in view of Corollary~\ref{corTTheta}, $\{{\bf T}_\theta\}_{\theta\in\DD}$ is bounded selfadjoint family of analytic operators and since $T_\theta$ is trace class ( see Theorem~\ref{ROODTh}) one has
$$
\spect_{\rm ess} {\bf T}_\theta=\spect_{\rm ac} {\bf T}_\theta=\spect T_0=[0,2^{-{1\over2}}].
$$
\subsection{Reduction by symmetries}
We use $(x,y)$ for the coordinates in $\R^2$ and we recall that $\Pi$ stands for the parity operator on $L^2(\R)$,  see Remark~\ref{remTTheta}(c). Let $\Pi_y:=\Pi\otimes1$, $\Pi_x:=1\otimes\Pi$ acting in $L^2(\R^2)$ denote respectively the reflection with respect to the $y$ and $x$ axis.
$H_\theta$ fulfills
$$
[H_\theta,\Pi_x]=[H_\theta,\Pi_y]=0.
$$
This allows to reduce $H_\theta$ as 
$$
H_\theta=\bigoplus_{\alpha,\beta\in\{\pm1\}} H_\theta^{\alpha,\beta},\quad {\rm where}\quad H_\theta^{\alpha,\beta}:=\Pi_x^\alpha\Pi_y^\beta H_\theta
$$
and $\Pi_x^\alpha:={1\over2}(\id+\alpha \Pi_x))$, $\Pi_y^\beta:={1\over2}(\id+\beta \Pi_y))$ denote the eigenprojectors of $\Pi_x$ and $\Pi_y$ resp.. We also stress that $H_\theta^{\alpha,\beta}$ is unitarily equivalent to the operator acting in $L^2(\R_+\times\R_+)$ with same symbol as $H_\theta$ but with additional Dirichlet boundary conditions  on $x=0$ if $\alpha=-1$ or on $y=0$ if $\beta=-1$ and Neumann boundary condition in the opposite case, see Table~\ref{tableSkeletons}. 

Similarly
${\bf T_\theta}$ enjoys the following symmetries
$$
[{\bf T_\theta},{\bf \Pi}]=[{\bf T_\theta},\EE]=[\Pi\oplus\Pi,\EE]=0.
$$
where ${\bf\Pi}:=\Pi\oplus\Pi:\SS\to\SS$ is the parity operator and $\EE:\SS\to\SS$ is the exchange of components operator: $\EE(\phi_1\oplus\phi_2):=\phi_2\oplus\phi_1$.  Thus
 we may consider separately
$$
{\bf \Pi}^\alpha\EE^\beta{\bf T}_\theta,\quad \alpha=\pm1,\quad\beta=\pm1 
$$
where ${\bf \Pi}^\alpha$ and $\EE^\beta$ denote the spectral projectors of ${\bf \Pi}$ and $\EE$ resp.:
$$
{\bf \Pi}^\alpha:={1\over2}(\id +\alpha{\bf \Pi})\quad
\EE^\alpha:={1\over2}(\id +\alpha\EE).
$$ 
One has the following elementary result the proof of which is left to the reader:
\begin{lemma}%
For all $\alpha$, $\beta$ in $\{\pm1\}$, ${\bf \Pi}^\alpha\EE^\beta{\bf T}_\theta$ is unitarily equivalent to 
$$
{\bf T}_\theta^{\alpha,\beta}:=\hat T_0+\beta \hat T_\theta^\alpha\quad\mbox{acting in $L^2(\R)$}.
$$
\end{lemma}%
It is then of practical importance to relate $H_\theta^{\alpha,\beta}$ and ${\bf T}_\theta^{\alpha,\beta}$.
\begin{lemma}\label{LemCorrespondenceBetweenHAndT}%
For all $\alpha$, $\beta$ in $\{\pm1\}$, $-k^2<-{1\over2}$ is a discrete eigenvalue of $H_\theta^{\alpha,\beta}$ iff $k>2^{-{1\over2}}$ is a discrete eigenvalue of ${\bf T}_\theta^{\alpha\beta,\beta}$. 
\end{lemma}%
\begin{proof} Due to the chosen orientation of the two normalized vectors, see Figure~\ref{cdr3Fig1}, we have the relations between mappings from $\HH^1(\R^2)$ to $L^2(\R)$
$$
\tau \Pi_y=\EE\tau,\quad \tau \Pi_x={\bf \Pi}\EE\tau
$$
so that
\begin{eqnarray*}
4\tau\Pi_x^\alpha \Pi_y^\beta&=& \tau(\id+\alpha\Pi_x)(\id+\beta\Pi_y)=(\id +\alpha {\bf \Pi}\EE)\tau(\id+\beta\Pi_y)\\
&=&(\id +\alpha {\bf \Pi}\EE)(\id + \beta\EE)\tau=(\id +\alpha\beta {\bf \Pi})(\id + \beta\EE)=4{\bf\Pi}^{\alpha\beta}\EE^\beta\tau.
\end{eqnarray*}
Let $R(z)^{\alpha,\beta}:=\Pi_x^\alpha\Pi_y^\beta(H_\theta-z)^{-1}$ and similarly for $R_0:=(H_0-z)^{-1}$ then using (\ref{KreinsFormula}) and Definition~\ref{SkeletonDef} we get
\begin{eqnarray*}
R(-k^2)^{\alpha,\beta}&=&R_0(-k^2)^{\alpha,\beta}-R_0(-k^2)^{\alpha,\beta}\tau^\star S(k)^{-1}\tau 
R_0(-k^2)^{\alpha,\beta}\\
&=&R_0(-k^2)^{\alpha,\beta}-R_0(-k^2)^{\alpha,\beta}\tau^\star \left(S(k)^{\alpha\beta,\beta}\right)^{-1}\tau 
R_0(-k^2)^{\alpha,\beta}
\end{eqnarray*}
with  $S(k)^{\alpha,\beta}:={\bf\Pi}^\alpha\EE^\beta S(k)$. The statement of the lemma now follows  easily.
\end{proof}
\begin{table}[ht]
\caption{}\label{tableSkeletons}
\renewcommand\arraystretch{1.5}
\noindent\[
\begin{array}{|c|c|c|}
\hline
 {\bf T}_\theta^{\alpha,\beta} & \mbox{subspace in $L^2(\R^2)$} & \mbox{B.C. on }  \R_+\times\R_+\ {\rm for}\ H_\theta\cr
\hline
 T_0+T_\theta^+& \Ran\Pi_x^+\Pi_y^+&  \mbox{N}_{x=0}\ {\rm N}_{y=0}\cr
 T_0-T_\theta^+& \Ran\Pi_x^-\Pi_y^+&   \mbox{D}_{x=0}\ {\rm N}_{y=0}\cr
 T_0+T_\theta^-& \Ran\Pi_x^-\Pi_y^-&   \mbox{D}_{x=0}\ {\rm D}_{y=0}\cr
 T_0-T_\theta^- & \Ran\Pi_x^+\Pi_y^-& \mbox{N}_{x=0}\ {\rm D}_{y=0}\cr
\hline
\end{array}
\]
\end{table}

\subsection{Existence {and monotonicity of }bound states}
First  we can quickly fix  two cases.
\begin{proposition}%
$H_\theta^{-,-}$ and $H_\theta^{-,+}$ have no discrete spectrum for all $\theta\in[\pi/2,\pi)$.
\end{proposition}%
\begin{proof} These two cases correspond to ${\bf T}_\theta^{\alpha,\beta}$ with  $\alpha\beta=-1$, see Lemma~\ref{LemCorrespondenceBetweenHAndT}. In view of (\ref{SignOfTTheta})  one has $\beta \hat T_\theta^\alpha\le0$ so that ${\bf T}_\theta^{\alpha,\beta}\le \hat T_0\le 2^{-{1\over2}}$, and therefore  ${\bf T}_\theta^{\alpha,\beta}$ cannot have an eigenvalue $k>2^{-{1\over2}}$. 
\end{proof}
\begin{proposition}%
(i) $H_\theta^{+,-}$ has no discrete spectrum for $\theta\in[{\pi\over2},{2\pi\over3}]$. (ii) It has at least one isolated eigenvalue for $\theta\in({2\pi\over3} ,\pi)$. The number of isolated eigenvalues of $H_\theta^{+,-}$ is bounded above  by $\|\widetilde T_\theta^-\|_1$ and therefore by the bound given in Lemma~ \ref{boundOnTiltdeTThetaMoins}. 
\end{proposition}%
\begin{proof} Here we have to consider ${\bf T}_\theta^{-,-}$. (i) We have using the operator $\widetilde T_\theta^-$, see (\ref{defTTildeThetaMoins}), that for all ${\pi\over2}\le \theta\le {2\pi\over3}$
\begin{eqnarray*}
2^{-{1\over2}}-{\bf T}_\theta^{-,-}
=2^{-{1\over2}}-\hat T_0+\hat T_\theta^-
=(2^{-{1\over2}}-\hat T_0)^{1\over2}(1+\widetilde T_{\theta}^-)(2^{-{1\over2}}-\hat T_0)^{1\over2}\ge0
\end{eqnarray*}
since $\widetilde T_\theta^-\ge-1$, see Remark~\ref{remarkTildeTTheta2PiOver3}(b) and (c).  Thus ${\bf T}_\theta^{-,-}\le 2^{-{1\over2}}$ which implies that it cannot have an eigenvalue larger than $2^{-{1\over2}}$.
To prove (ii) it is sufficient  to show that there exists $k>2^{-{1\over2}}$ so that $\inf \spect(k-{\bf T}_\theta^{-,-})<0$ for all $\theta>2\pi/3$. 
We first establish that the lowest eigenvalue $\widetilde E_1^-(\theta)$ of $\widetilde T_\theta^-$ is strictly smaller than $-1$ for all $2\pi/3<\theta<\pi$. Let $\varphi$ denote  the normalized    eigenvector of $\widetilde T_{2\pi\over3}^-$ associated to the eigenvalue $-1$, see (\ref{exactPbVp3}).  By the Feynman Hellman theorem one has
\begin{eqnarray*} 
{d\over d\theta}\widetilde E_1^-\left(2\pi\over3\right)
&=& \left({d\over d\theta}\widetilde T_\theta^-\varphi,\varphi\right)_{|_{\theta={2\pi\over 3}}}\\
&=&
{1\over \frac{1}{18} \left(6-\sqrt{3} \pi \right)}
\int_\R\left(\partial_\theta T_\theta^{-}(p,q)  \right)_{|_{\theta={2\pi\over3}}}
\frac{1}{p \left(2 p^2+3\right)} \frac{1}{q \left(2 q^2+3\right)}dpdq\\
&=&-\frac{\pi }{2 \left(6-\sqrt{3} \pi \right)}\sim-2.81201.
\end{eqnarray*} 
This shows that in a right neighbourhood of $2\pi/3$ we have that $\widetilde E_1^-(\theta)<-1$. Thanks to Remark~\ref{remarkTildeTTheta2PiOver3}(c), this remains true for all $\theta>2\pi/3$. Let $A_k:=(k-\hat T_0)^{1\over2}$ then
$$
k-{\bf T}_\theta^{-,-}=A_k(1+\widetilde T_\theta^-(k)) A_k,\quad{\rm with}\quad
\widetilde T_\theta^-(k):={A_{2^{-{1\over2}}}\over A_k} \widetilde T_\theta ^-{A_{2^{-{1\over2}}}\over A_k}.
$$
Clearly $\widetilde T_\theta^-(k)$ converges in norm to $\widetilde
T_\theta^-$ as $k\to2^{-{1\over2}}$. Since for all $\theta>2\pi/3$
there exists $a>0$ so that $\inf \spect(\widetilde T_\theta^-)\le
-1-a$, one can find $k$ close enough to $2^{-{1\over2}}$ so that 
$\inf\spect(\widetilde T_\theta^-(k))\le -1-a/2$, i.e. there exists
$\varphi\in L^2(\R)$ such that $((1+\widetilde
T_\theta^-(k))\varphi,\varphi)\le -a/2\|\varphi\|^2$. Let
$\psi:=A_k^{-1}\varphi$; we finally get that
$$
(k-{\bf T}_\theta^{-,-}\psi,\psi)\le -{a\over2}\| A_k\psi\|^2\le -{(k-2^{-{1\over2}})a\over2}\|\psi\|^2
$$
which shows that $\inf \spect(k-{\bf T}_\theta^{-,-})<0$. The bound on the number of isolated eigenvalues is standard, see e.g. the proof of Theorem~3.3 in \cite{BEKS}.
 \end{proof}
\begin{proposition}%
$H_\theta^{+,+}$ has at least one isolated eigenvalue for all $\theta\in(0,\pi)$ and this eigenvalue is unique in $[\pi/3,2\pi/3]$. 
\end{proposition}%
\begin{proof} Here we have to consider ${\bf T}_\theta^{+,+}$.
Take a smooth even function $j\in C_0^\infty(\R,\R_+^\star)$,
such that $\int_{\R} j(x) dx=1$. For every
$\varepsilon>0$ define $\psi_\varepsilon(x)=(1/\varepsilon) j(x/\varepsilon)$ and 
$$
\phi_\varepsilon := \frac{\sqrt{\varepsilon}}{||j||}\psi_\varepsilon.
$$
We have $||\phi_\varepsilon||=1$, while $\psi_\varepsilon$ is an
approximation of Dirac's distribution. An elementary calculus gives
$$
(\hat T_0\psi_\varepsilon,\psi_\varepsilon)={1\over\sqrt{2}}+\OO(\varepsilon^2),\quad (\hat T_\theta^+\psi_\varepsilon,\psi_\varepsilon)=\varepsilon \hat T_\theta^+(0,0)+\OO(\varepsilon^2)
$$
and since $\hat T_\theta^+(0,0)=\hat T_\theta(0,0)=\pi|\sin(\theta)|/2>0$ it follows that ( by taking $\varepsilon>0$ small enough)
$$
\sup_{\psi\in L^2(\R_+)} ({\bf T}_\theta^{+,+}\psi,\psi)>{1\over\sqrt{2}}=\sup\spect_{\rm ess} {\bf T}_\theta^{+,+}.
$$
Therefore ${\bf T}_\theta^{+,+}$ has at least one eigenvalue larger than
$1\over\sqrt{2}$.  Thanks to the monotonicity of ${^\sharp\bf T}^{+,+}_\theta$ ( see below) and its symmetry w.r.t $\pi/2$,  it sufficient to show the uniqueness, that ${^\sharp\bf T}^{+,+}_{2\pi/3}$ or equivalently $H_{2\pi/3}^{+,+}$ has at most one isolated eigenvalue. On the other hand $H_{2\pi/3}$ is bounded below by $H$ of (\ref{RescaledThreePartHamiltonian}), with $\lambda=-1$, which is shown to possess a unique bound state by Mc Guire in \cite[\S IV. D]{McG}.
\end{proof}

To prove the monotonicity, we remark that 
${\bf T}_\theta^{\alpha,\alpha}$ is unitarily equivalent with 
${^\sharp{\bf T}}_\theta^{\alpha,\alpha}:= 
(2+p^2\sin^2(\theta))^{-{1\over2}}+\alpha{^\sharp T}_\theta^\alpha$, 
see (\ref{defTThetaSharp}). Since $\pm {^\sharp T}_\theta^\pm\ge0$ for
all $\theta\in[\pi/2,\pi)$, and thanks to the explicit and simple
dependence on $\theta$ of ${^\sharp{\bf T}}_\theta^{\alpha,\alpha}$,
we infer that both families, $\alpha=\pm1$, are monotonously
increasing as functions of $\theta$, and therefore their eigenvalues
$k$ have the same property. It follows that $-k^2$, the eigenvalues of $H_\theta$ are decreasing. Using the bounds $\|T_\theta\|\le 2^{-{1\over2}}$, see Remark~\ref{remTTheta}(a), it follows that $\| {\bf T}_\theta^{\alpha,\alpha}\|\le \sqrt{2}$ and consequently $H_\theta\ge -2$. 
We gather in a final theorem what we have established so far. Clearly this problem has two mirror symmetries. Following \cite{BEPS} we shall call {\sl axis of the scissor} the one which lies in the smaller angle, i.e. the $x$ axis with our notations and {\sl second axis of the scissor} the other one. We warn the reader that our $\theta$ is not the one of \cite{BEPS}. 

\begin{theorem}\label{thQS}%
(i) The Hamiltonian $H_\theta$ has no isolated bound state which is odd with respect to the axis of the scissor. 
\vskip1mm
\nid 
(ii) It has no isolated bound state which is odd with respect to the second axis of the scissor when $\pi/2\le\theta\le 2\pi/3$. It has at least one isolated bound state which is odd with respect to this second axis when $2\pi/3<\theta<\pi$. The number of such bound states is bounded above by  $\|\widetilde T_\theta^-\|_1$, and therefore by the bound given in Lemma~ \ref{boundOnTiltdeTThetaMoins}.
\vskip1mm
\nid 
(iii) It has 
at least one bound state for all $0<\theta<\pi$ which is even with respect to both axis of the scissor and unique for $\theta$ in $[\pi/3,2\pi/3]$. 
\vskip1mm
\nid 
{ (iv) All bound states of $H_\theta$ are  bounded below by $-2$ and monotonously 
decreasing with respect to $\theta$ on $[\pi/2,\pi)$.}
\end{theorem}%

\section{Concluding remarks and open problems}
We are far from having found the answers to  all the questions about the quantum scissor of \S~\ref{sectionQS}. Let us review these questions; most of them are already in \cite[\S III]{BEPS}: 
\begin{enumerate}
\item Every bound state is even w.r.t. the scissor axis: {\em done, see Th~\ref{thQS}(i)}.
\item With respect to the second axis the bound states can have both parities ({\em done, see Th~\ref{thQS}(ii) and (iii)}) which are alternating if the bound states are arranged according to their energies: {\em not done}. 
\item As the angle $\theta$ gets larger new bound states emerge from the continuum. Fnd the correponding critical values $\theta_c$ of $\theta$: {\em very partially done}.    Thanks to Th~\ref{thQS}(ii,iii) we know that the first 
critical value of $\theta$ is $2\pi/3$.  Compute the asymptotic of the number of bound state as $\theta\to\pi$:  {\em not done}.
\item All the bound state energies are monotonically decreasing
  functions of $\theta$: {\em {done}}.
\item  Do we have a bound state or a resonance at the threshold, 
when a bound state emerges from the continuum? {\em Not done}.
\item The other way around, when $\theta$ decreases, do the bound
  states become resonances? Can we follow them? {\em Not done}.
\item Can one expand the new bound state emerging from the continuum as a function of $\theta-\theta_c$? {\em Not done}.
\end{enumerate}

Concerning the integral operator $T_\theta$, see \S~\ref{SectionTTheta}, 
\begin{enumerate}
\item can we enlarge the list of exact spectral results, see \S~\ref{ExactEigenvaluesAndEigenvectorsSect}? 
{\item Are all eigenvalues of $T_\theta^\pm$ simple?}
\item { Is it true that $\theta\to T_\theta^\pm$} are monotonous?
\end{enumerate}


\bibliographystyle{amsalpha}%

\end{document}